\documentclass[a4paper,UKenglish,cleveref, autoref, thm-restate]{lipics-v2021}

\pdfoutput=1 
\hideLIPIcs  

\usepackage{tikz-cd}

\newcommand{\Q}{\mathbb{Q}}

\newcommand{\N}{\mathbb{N}}

\newcommand{\Z}{\mathbb{Z}}
\def\crys{{\text{\upshape cris}}}
\newcommand*{\lbra}{\mathopen{[\![}}
\newcommand*{\rbra}{\mathclose{]\!]}}
\def\coeff{\operatorname{coeff}\nolimits}
\def\eval{\operatorname{eval}\nolimits}

\theoremstyle{definition}
\newtheorem{examples}[theorem]{Examples}

\usepackage{color}
\definecolor{lightblue}{rgb}{0.8,0.85,1}
\usepackage{listings}
\definecolor{keywordcolor}{rgb}{0.7, 0.1, 0.1}   
\definecolor{tacticcolor}{rgb}{0.0, 0.1, 0.6}    
\definecolor{commentcolor}{rgb}{0.4, 0.4, 0.4}   
\definecolor{symbolcolor}{rgb}{0.0, 0.1, 0.6}    
\definecolor{sortcolor}{rgb}{0.1, 0.5, 0.1}      
\definecolor{attributecolor}{rgb}{0.7, 0.1, 0.1} 

\newcommand{\code}[1]{\lstinline{#1}}

\newcommand{\extlink}{~\ensuremath{{}^\text{\faExternalLink*}}}

\title{A Formalization of Divided Powers in Lean} 

\author{Antoine {Chambert-Loir}\footnotemark[1]}{Université Paris Cité, Institut de mathématiques de Jussieu -- Paris rive gauche, France \and \url{https://webusers.imj-prg.fr/~antoine.chambert-loir/} }{antoine.chambert-loir@u-paris.fr}{https://orcid.org/0000-0001-8485-7711}{The author acknowledges the support of the APRAPRAM \'Emergence project of Université Paris Cité.}

\author{María Inés {de Frutos-Fernández}\footnotemark[1]}{Universit\"at Bonn, Germany \and \url{https://mariainesdff.github.io/}}{midff@math.uni-bonn.de}{https://orcid.org/0000-0002-5085-7446}{The author was funded by the Deutsche Forschungsgemeinschaft (DFG, German
	Research Foundation) under Germany's Excellence Strategy -- EXC-2047/1 --
	390685813, and by the Grant
	CA1/RSUE/2021-00623 funded by the Spanish Ministry of
	Universities, the Recovery, Transformation and Resilience
	Plan, and Universidad Autónoma de Madrid. Part of this work was conducted during the HIM
	trimester programme Prospects of Formal Mathematics, also funded by EXC-2047/1 – 390685813.}

\authorrunning{A. Chambert-Loir and M.\,I. de Frutos-Fernández} 

\Copyright{Antoine Chambert-Loir and María Inés de Frutos-Fernández} 



\category{} 

\relatedversion{} 

\supplement{}

\supplementdetails[subcategory={}, cite={}, swhid={}]{Software}{https://github.com/mariainesdff/divided_powers_journal}

\funding{This work was started during the ICERM workshop `Lean for the Curious Mathematician 2022', funded by the US National Science Foundation under Grant No. DMS-1929284.}

\acknowledgements{We are grateful to Kevin Buzzard for suggesting this project and for his ongoing support during its completion.
We would also like to thank Anatole Dedecker, Floris van Doorn and Bhavik Mehta for interesting discussions about the project. Finally, we thank the Mathlib community and reviewers for useful feedback on our contributions to the library.}

\nolinenumbers 

\EventEditors{Yannick Forster and Chantal Keller}
\EventNoEds{2}
\EventLongTitle{16th International Conference on Interactive Theorem Proving (ITP 2025)}
\EventShortTitle{ITP 2025}
\EventAcronym{ITP}
\EventYear{2025}
\EventDate{September 29--October 2, 2025}
\EventLocation{Reykjavik, Iceland}
\EventLogo{}
\SeriesVolume{352}
\ArticleNo{4}

\begin{document}

\maketitle

\begin{abstract}
Given an ideal $I$ in a commutative ring $A$, a divided power structure on $I$ is a collection of maps $\{\gamma_n \colon I \to A\}_{n \in \mathbb{N}}$,
subject to axioms that imply that it behaves like the family $\{x \mapsto \frac{x^n}{n!}\}_{n \in \mathbb{N}}$, but which can be defined even when division by factorials is not possible in $A$. Divided power structures have important applications in diverse areas of mathematics, including algebraic topology, number theory and algebraic geometry.

In this article we describe a formalization in Lean 4 of the basic theory of divided power structures, including divided power morphisms and sub-divided power ideals, and we provide several fundamental constructions, in particular quotients and sums. This constitutes the first formalization of this theory in any theorem prover.

As a prerequisite of general interest, we expand the formalized theory of multivariate power series rings, endowing them with a topology and defining evaluation and substitution of power series.
\end{abstract}

\footnotetext[1]{Both authors contributed equally to this research.}

\section{Introduction}
\label{sec:intro}

The theory of divided powers, invented by Cartan in~\cite{Car},
grew out of the observation that there are algebraic contexts where
one can make sense of ratios of the form $h^k/k!$, even though division by integers such as $k!$ is not a priori possible in an arbitrary ring. Cartan also found out the relevant functional equations satisfied by these ratios.

The context of Cartan's work is algebraic topology.
In~\cite{Car}, the elements~$h$ are classes of even degree 
in a graded commutative Eilenberg--Maclane algebra.
While such algebras are not commutative,
the part where the divided powers take place is essentially commutative,
and the theory has not been elaborated further
for noncommutative rings.
On the other hand,  the case of commutative rings  has been the object
of a detailed inquiry, 
starting from the works of Roby~\cite{Roby63,Roby65,Roby66}. For this reason, here we only consider commutative rings.

In the 1960s, under the impulse of Grothendieck, algebraic geometers tried to understand the relations between the various cohomology theories that could be defined on algebraic varieties. It was soon observed that the ``algebraic de Rham cohomology'' led to technical difficulties, essentially because the primitive of a polynomial \( \sum a_n T^n\) is
given by the expression \(\sum \frac1{n+1} a_n T^{n+1}\), which involves integer denominators.
One of the suggestions made by Grothendieck was to work, not with polynomials, but in a ``divided power algebra'', in which all ratios \(T^n/n!\) make sense.
This suggestion was followed by Berthelot~\cite{Ber74} for his development of crystalline cohomology. 

Despite their complexity, crystalline cohomology and subsequent $p$-adic cohomologies are fundamental to modern arithmetic geometry. 
In particular, they are one of the multiple tools used by Wiles
and his collaborators in the proof of Fermat's Last Theorem.
The present formalization of divided powers
is thus one step in the formalization of this theorem {(see Section \ref{par:Hodge} for more technical details)}. 

In this article we formalize in Lean 4 the basic theory of divided power rings, divided power morphisms and sub-divided power ideals. We also provide several constructions of divided powers, including quotients and sums. To the authors' knowledge, this is the first formalization of the theory of divided powers in any theorem prover.

As a prerequisite, which has applications far beyond the scope of this project, we formalize the topological theory of multivariate power series rings and make it possible to evaluate and substitute power series.

Our work constitutes the first formalization of these topological aspects of multivariate power series, but related work exists in several theorem provers. The ring of multivariate power series (as an abstract ring) has been formalized in Mizar \cite{CAMizar}, while the topological ring of univariate power series has been formalized in Isabelle/HOL \cite{FPSIsa}. In Rocq, a formalization of the abstract ring of univariate power series is available in the UniMath library \cite{padicsRocq}, and a separate implementation built over the MathComp library is described in \cite{ACRocq}. Note that none of these include a formalization of evaluation or substitution of univariate or multivariate formal power series.

\subsection{Lean and Mathlib}
\label{subsec:lean}
This formalization project uses the theorem prover Lean \cite{Lean3, Lean4} and builds over Lean's mathematical library Mathlib \cite{Mathlib}, which was ported to the latest version of the language (Lean 4) in 2023. While the project started in Lean 3 and the first Mathlib contributions from it were made to Mathlib3, after the library port we translated our existing work to Lean 4 and used the latter version for the rest of the development.

Throughout the text, we use the clickable symbol {\faExternalLink*} to link to the source code for the formalization. The original work described in this article comprises about 8,000 lines of code, of which approximately 5,000 have been merged into Mathlib; the rest of our code can be found in the {public GitHub repository \code{mariainesdff/divided_powers_journal}\href{https://github.com/mariainesdff/divided_powers_journal}{\extlink}}. We keep a record of past and ongoing contributions to Mathlib arising from this work in the Lean prover community project `Divided Powers' \href{https://github.com/orgs/leanprover-community/projects/25/views/1}{\extlink}.

\subsection{Paper Outline}
\label{subsec:outline}
In \S\ref{sec:dp}, we describe the general theory of divided power structures: after introducing the definition and some key examples, we treat divided power morphisms in \S\ref{subsec:dpmorph}, sub-divided power ideals in \S\ref{subsec:sub-dpi}, and divided powers on quotient rings in \S\ref{subsec:quot}. \S\ref{sec:exp} is devoted to the formalization of the module of exponential power series: in \S\ref{subsec:topology} and \S\ref{subsec:linear-top} we describe the topological theory of multivariate power series rings, which is needed in \S\ref{subsec:eval-subst} to define evaluation and substitution operations and in 
\S\ref{subsec:exp} to describe the exponential module, which is later used in \S\ref{subsec:ideal-add} to define divided powers on sums of ideals. Finally, we conclude with some implementation remarks in \S\ref{subsec:impl} and a description of future work in \S\ref{subsec:future}.

\section{Divided Power Structures}\label{sec:dp}

\subsection{Definition}\label{subsec-dp-def}

\begin{definition}[{\cite[Definition 3.1]{BO78}}]\label{def:dp}
Let $I$ be an ideal in a commutative ring $A$.
A \emph{divided power structure on $I$} is a family of maps $\gamma_n : I \to A$ for $n \in \mathbb{N}$ such that
\begin{enumerate}[(i)]
	\item $\forall x \in I$, $\gamma_0 (x) = 1$.
	\item $\forall x \in I$, $\gamma_1 (x) = x$.
	\item $\forall x \in I, \forall n > 0, \gamma_n (x) \in I$.
	\item $\forall x, y \in I, \gamma_n (x + y) = \sum_{i + j=n} {\gamma_i(x) \cdot \gamma_j(y)}$.
	\item $\forall a \in A, \forall x\in I, \gamma_n (a \cdot x) = a^n \cdot \gamma_n(x)$.
	\item $\forall x \in I, \forall m, n \in \mathbb{N}, \gamma_m (x) \cdot \gamma_n (x) = \binom{m + n}m \cdot \gamma_{m + n} (x)$.
	\item $\forall x \in I, \forall m \in \mathbb{N}, \forall n > 0, \gamma_m (\gamma_n (x)) = \frac{(m \cdot n)!}{m!(n!)^m} \cdot \gamma_{m \cdot n} (x)$.
\end{enumerate}
If $I$ is an ideal equipped with a divided power structure, we say that $I$ \emph{has divided powers}, that $(I, \gamma)$ is a \emph{divided power ideal}, and that $(A, I, \gamma)$ is a \emph{divided power ring} or a \emph{divided power algebra}. 
\end{definition}

\begin{lstlisting}[caption={Definition of divided power structure.},float=h, label=code:dp,captionpos=t,abovecaptionskip=-\medskipamount]
structure DividedPowers {A : Type*} [CommSemiring A] (I : Ideal A) where
  dpow : ℕ → A → A
  dpow_null : ∀ {n x} (_ : x ∉ I), dpow n x = 0
  dpow_zero : ∀ {x} (_ : x ∈ I), dpow 0 x = 1
  dpow_one : ∀ {x} (_ : x ∈ I), dpow 1 x = x
  dpow_mem : ∀ {n x} (_ : n ≠ 0) (_ : x ∈ I), dpow n x ∈ I
  dpow_add : ∀ (n) {x y} (_ : x ∈ I) (_ : y ∈ I),
    dpow n (x + y) = (antidiagonal n).sum fun k ↦ dpow k.1 x * dpow k.2 y
  dpow_mul : ∀ (n) {a : A} {x} (_ : x ∈ I),
    dpow n (a * x) = a ^ n * dpow n x
  mul_dpow : ∀ (m n) {x} (_ : x ∈ I),
    dpow m x * dpow n x = choose (m + n) m * dpow (m + n) x
  dpow_comp : ∀ (m) {n x} (_ : n ≠ 0) (_ : x ∈ I),
    dpow m (dpow n x) = uniformBell m n * dpow (m * n) x
\end{lstlisting}

When $A$ is a $\Q$-algebra, one may set $\gamma_n(x)=x^n/n!$
for all $x\in I$ and these axioms are easy to check.
In particular, axiom~(iv) is a version of the binomial theorem without binomial coefficients.

\begin{remark}
The fraction $\frac{(m \cdot n)!}{m!(n!)^m}$ appearing in condition (vii) is a natural number for all $n, m \in \N$; we formalize this definition as \code{Nat.uniformBell}\footnote[2]{The terminology comes from that of Bell numbers in combinatorics,
	which represent the number of partitions
	of a set of given cardinality; here, we fix one particular partition type.}\href{https://github.com/leanprover-community/mathlib4/blob/951bd16d55870f47b2e2cd98c794f933e29bd2b5/Mathlib/Combinatorics/Enumerative/Bell.lean#L121C5}{\extlink}. This quantity has a combinatorial significance: it counts the number of ways to divide a set with $m \cdot n$ elements into $m$ groups of $n$-element subsets.

\end{remark}

We formalize Definition \ref{def:dp} as a structure \lstinline{DividedPowers}\href{https://github.com/leanprover-community/mathlib4/blob/951bd16d55870f47b2e2cd98c794f933e29bd2b5/Mathlib/RingTheory/DividedPowers/Basic.lean#L73}{\extlink} (see Listing \ref{code:dp}), where the field \lstinline{dpow} represents the family of functions $\{\gamma_n\}_{n \in \N}$. Note that, instead of defining \lstinline{dpow} as a function $\N \to I \to A$, we define it as a function $\N \to A \to A$ and impose in the field \lstinline{dpow_null} the condition that this map is zero outside the ideal $I$ (see \S\ref{subsubsec:dpow} for a discussion of this choice).

In our formalized definition, we only require $A$ to be a commutative semiring instead of a commutative ring; that is, we do not assume that $A$ has a negation operator. However, as mentioned in the introduction, it does not seem that a general,
and relevant, definition of divided powers has been given
in the noncommutative setting,
so we do not relax the commutativity hypothesis.
 
In the field \lstinline{dpow_add}, the expression \lstinline{(antidiagonal n).sum fun k ↦ dpow k.1 x * dpow k.2 y} represents the sum $\sum_{k_1 + k_2 =n} {\gamma_{k_1}(x) \cdot \gamma_{k_2}(y)}$\href{https://github.com/leanprover-community/mathlib4/blob/951bd16d55870f47b2e2cd98c794f933e29bd2b5/Mathlib/Algebra/Order/Antidiag/Prod.lean#L55}{\extlink}, while in \lstinline{dpow_comp}, \lstinline{uniformBell m n} represents the fraction $\frac{(m \cdot n)!}{m!(n!)^m}$.

\subsection{Examples of Divided Power Structures}\label{subsec-dp-ex}
We formalized several examples of divided power structures, starting with the trivial one.
\begin{example}\label{ex:dp-zero}
In any commutative ring $A$, the trivial ideal $(0)$ is a divided power ideal with respect to the family $\gamma_n(0) = 0$ for all $n \ge 1$\href{https://github.com/leanprover-community/mathlib4/blob/951bd16d55870f47b2e2cd98c794f933e29bd2b5/Mathlib/RingTheory/DividedPowers/Basic.lean#L91}{\extlink}.
\end{example}

We next consider several ideals for which the family $\gamma_n(x) = \frac{x^n}{n!}$ is a divided power structure. We start by making the definition \lstinline{OfInvertibleFactorial.dpow}\href{https://github.com/mariainesdff/divided_powers_journal/blob/a8c64b021200b8c2cee79296881b5c0a3379092a/DividedPowers/RatAlgebra.lean#L77}{\extlink}, where we represent $\frac{x^n}{n!}$ by 
\lstinline{Ring.inverse (n ! : A) * x ^ n}. Note that \lstinline{Ring.inverse} is defined for any ring $A$ as the function sending \lstinline{x : A} to its inverse, if it exists, or to $0$, otherwise. Hence we can write down this definition without any extra hypothesis on $A$ or $I$. However, for it to define a divided power structure, some extra conditions are required.

\begin{lstlisting} 
variable {A : Type*} [CommSemiring A] (I : Ideal A) 

def OfInvertibleFactorial.dpow : ℕ → A → A := 
  fun n x ↦ if x ∈ I then Ring.inverse (n ! : A) * x ^ n else 0 
\end{lstlisting} 

\begin{example}\label{ex:dp-rat}
If $A$ is a $\Q$-algebra, then any ideal in $A$ has a unique divided power structure, given by the family $\gamma_n(x) = \frac{x^n}{n!}$\href{https://github.com/mariainesdff/divided_powers_journal/blob/78f8956f06816221781a16b5223753c97b9b641c/DividedPowers/RatAlgebra.lean#L273}{{\extlink}}.
\end{example}

Another sufficient condition for $\gamma_n(x) = \frac{x^n}{n!}$ to be a divided power structure on an ideal $I$ is the existence of a natural number $n$ such that $(n-1)!$ is invertible in $A$ and $I^n = 0$. We call
this divided power structure \lstinline{OfInvertibleFactorial.dividedPowers}\href{https://github.com/mariainesdff/divided_powers_journal/blob/78f8956f06816221781a16b5223753c97b9b641c/DividedPowers/RatAlgebra.lean#L193}{\extlink}.

\begin{lstlisting}
def OfInvertibleFactorial.dividedPowers {n : ℕ}
    (hn_fac : IsUnit ((n - 1).factorial : A)) (hnI : I ^ n = 0) : 
    DividedPowers I where
  dpow := OfInvertibleFactorial.dpow I
  ...
\end{lstlisting}

In particular, we can specialize this definition to obtain three interesting examples.
\begin{examples}\label{ex:dp-nil}
	\begin{enumerate}
		\item If $A$ is any commutative ring and $I \subseteq A$ is any ideal such that $I^2 = 0$, then 
		$\gamma_n(x) = \frac{x^n}{n!}$ is a divided power structure on $I$, because $(2-1)! = 1$ is always invertible\href{https://github.com/mariainesdff/divided_powers_journal/blob/78f8956f06816221781a16b5223753c97b9b641c/DividedPowers/RatAlgebra.lean#L217}{\extlink}.
		\item If $p$ is a prime number which is nilpotent in $A$ (that is, $p^n = 0$ in $A$ for some natural number $n$), and $I^p = 0$, then $\gamma_n(x) = \frac{x^n}{n!}$ is a divided power structure on $I$, since $(p-1)!$ is a unit in $A$\href{https://github.com/mariainesdff/divided_powers_journal/blob/78f8956f06816221781a16b5223753c97b9b641c/DividedPowers/RatAlgebra.lean#L235}{\extlink}.
		\item Recall that the \emph{characteristic} of $A$ is the smallest positive natural number $p$ such that $p$ copies of $1 \in A$ add up to $0 \in A$, if such a number exists, or zero otherwise. 
		If the characteristic of $A$ is a prime number $p$ and $I^p = 0$, then $\gamma_n(x) = \frac{x^n}{n!}$ is a divided power structure on $I$. This follows as a special case of the previous example, since the characteristic of a ring is a nilpotent element\href{https://github.com/mariainesdff/divided_powers_journal/blob/78f8956f06816221781a16b5223753c97b9b641c/DividedPowers/RatAlgebra.lean#L253}{\extlink}.

	\end{enumerate}
\end{examples}

Finally, we consider an example which is of key importance in number theory. For each prime number $p$, we can define a \emph{$p$-adic absolute value} on the rational numbers; the completion of $\Q$ with respect to this absolute value is the field $\Q_p$ of \emph{$p$-adic numbers}. The subring of $\Q_p$ consisting of elements with absolute value at most one is the ring $\Z_p$ of \emph{$p$-adic integers}. The ring $\Z_p$ has a unique maximal ideal, generated by $p$; this ideal consists of the elements of absolute value strictly smaller than one.

\begin{example}\label{ex:padic}
The family $\gamma_n(x) = \frac{x^n}{n!}$ is a divided power structure on the $\Z_p$-ideal~$(p)$\href{https://github.com/mariainesdff/divided_powers_journal/blob/78f8956f06816221781a16b5223753c97b9b641c/DividedPowers/Padic.lean#L213}{\extlink}. The key part of the proof is to show that, for every $x \in (p)$ and every nonzero~$n$, the fraction $\gamma_n(x) = \frac{x^n}{n!}$ has $p$-adic absolute value less than one, and hence belongs to the ideal~$(p)$.
Then, the verification of the axioms of a divided power structure can be checked in~$\Q_p$, and they follow from Example \ref{ex:dp-rat} since $\Q_p$ is a $\Q$-algebra. 
\end{example}

We formalize Example \ref{ex:padic} as a special case of a more general construction: if $f : S \to A$ is a ring homomorphism and $J$ is an $S$-ideal such that $\text{span}(f(I)) = J$ and such that $f^{-1}(\gamma_n(f(x)) \in I$ for all $n \ne 0, x \in I$, then $f$ induces a divided power structure on $I$\href{https://github.com/mariainesdff/divided_powers_journal/blob/78f8956f06816221781a16b5223753c97b9b641c/DividedPowers/Padic.lean#L70}{\extlink}.

\subsection{Divided Power Morphisms}
\label{subsec:dpmorph}
Let $(A, I, \gamma)$ and $(B,J,\delta)$ be two divided power rings.

\begin{definition}[{\cite[\S3]{BO78}}]\label{def:dpmorph}
A \emph{divided power morphism} $f : (A, I, \gamma) \to (B,J,\delta)$ is a ring homomorphism $f : A \to B$ such that $f(I) \subseteq J$ and $\delta_n (f(x)) = f (\gamma_n(x))$ for all $n \in \mathbb{N}, x \in I$.
\end{definition}

We provide two formalized versions of this definition: a propositional one, which we call \lstinline{IsDPMorphism}\href{https://github.com/leanprover-community/mathlib4/blob/951bd16d55870f47b2e2cd98c794f933e29bd2b5/Mathlib/RingTheory/DividedPowers/DPMorphism.lean#L58}{\extlink}, as well as a bundled version \lstinline{DPMorphism}\href{https://github.com/leanprover-community/mathlib4/blob/951bd16d55870f47b2e2cd98c794f933e29bd2b5/Mathlib/RingTheory/DividedPowers/DPMorphism.lean#L97}{\extlink} which contains both the map $f : A \to B$ and the properties it must satisfy.

\begin{lstlisting}
structure DividedPowers.IsDPMorphism (f : A →+* B) : Prop where
  ideal_comp : I.map f ≤ J
  dpow_comp : ∀ {n}, ∀ a ∈ I, hJ.dpow n (f a) = f (hI.dpow n a)
  
structure DividedPowers.DPMorphism extends RingHom A B where
  ideal_comp : I.map toRingHom ≤ J
  dpow_comp : ∀ {n}, ∀ a ∈ I, hJ.dpow n (toRingHom a) = toRingHom (hI.dpow n a)
\end{lstlisting}

While the unbundled definition is more convenient to develop the basic theory of divided power morphisms, we anticipate that the bundled version will be useful for future work towards formalizing the crystalline site, so we provide it as well.

We formalize several fundamental results about divided power morphisms, including the following propositions.

\begin{proposition}[{\cite[Proposition 2]{Roby65}}]\label{prop:roby65-prop2}
Let 
$f : A \to B$ be a ring homomorphism such that $f(I) \subseteq J$. Then the set of elements $x \in A$ such that $\delta_n (f(x)) = f (\gamma_n(x))$ for all $n \in \N$ is an ideal of $A$\href{https://github.com/leanprover-community/mathlib4/blob/951bd16d55870f47b2e2cd98c794f933e29bd2b5/Mathlib/RingTheory/DividedPowers/DPMorphism.lean#L134}{\extlink}.
\end{proposition}

\begin{proposition}[{\cite[Proposition 3]{Roby65}}]\label{prop:roby65-prop3}
Let $S$ be a generating set for the ideal $I$. 
If $f : A \to B$ is a ring homomorphism such that $f(I) \subseteq J$, and for every $n \in \N$ and $x \in S$ the equality $\delta_n (f(x)) = f (\gamma_n(x))$ holds, then $f$ is a divided power morphism\href{https://github.com/leanprover-community/mathlib4/blob/951bd16d55870f47b2e2cd98c794f933e29bd2b5/Mathlib/RingTheory/DividedPowers/DPMorphism.lean#L159}{\extlink}.
\end{proposition}

\begin{corollary}\label{cor:roby65-prop3}
Let $S$ be a generating set for the ideal $I$, and let $\gamma'$ be a second divided power structure on $I$. If $\gamma'_n (x) = \gamma_n(x)$ for all $n \in \N$ and $x \in S$, then $\gamma' = \gamma$.
\end{corollary}
\begin{proof}
Apply Proposition \ref{prop:roby65-prop3} with $f$ being the identity map on $A$ and $\delta := \gamma'$\href{https://github.com/leanprover-community/mathlib4/blob/951bd16d55870f47b2e2cd98c794f933e29bd2b5/Mathlib/RingTheory/DividedPowers/DPMorphism.lean#L239}{\extlink}.
\end{proof}

\subsection{Sub-Divided Power Ideals}
\label{subsec:sub-dpi}
In this section we describe the theory of sub-dp-ideals, following \cite[\S5]{Roby65} and 
\cite[\S3]{BO78}.

\begin{definition}
Let $(I, \gamma)$ be a divided power ideal of $A$. A subideal $J \le I$ is called a \emph{sub-divided power ideal}, or \emph{sub-dp-ideal} of $I$ if $\gamma_n(x) \in J$ for all $n > 0, x \in J$.
\end{definition}

As we did for divided power morphisms, we provide an unbundled and a bundled version of this definition, respectively called \lstinline{IsSubDPIdeal}\href{https://github.com/mariainesdff/divided_powers_journal/blob/78f8956f06816221781a16b5223753c97b9b641c/DividedPowers/SubDPIdeal.lean#L127}{\extlink} and \lstinline{SubDPIdeal} \href{https://github.com/mariainesdff/divided_powers_journal/blob/78f8956f06816221781a16b5223753c97b9b641c/DividedPowers/SubDPIdeal.lean#L275}{\extlink}.
\begin{lstlisting}
structure IsSubDPIdeal (J : Ideal A) : Prop where
  isSubideal : J ≤ I
  dpow_mem : ∀ {n : ℕ} (_: n ≠ 0) {j : A} (_ : j ∈ J), hI.dpow n j ∈ J
  
structure SubDPIdeal where
  carrier : Ideal A
  isSubideal : carrier ≤ I
  dpow_mem : ∀ {n : ℕ} (_ : n ≠ 0), ∀ j ∈ carrier, hI.dpow n j ∈ carrier
\end{lstlisting}

If $J$ is a sub-dp-ideal of $(I, \gamma)$, the family of maps obtained by restricting $\gamma$ to $J$ is a divided power structure on $J$\href{https://github.com/mariainesdff/divided_powers_journal/blob/78f8956f06816221781a16b5223753c97b9b641c/DividedPowers/SubDPIdeal.lean#L144}{\extlink}.
\begin{lstlisting}
def IsSubDPIdeal.dividedPowers {J : Ideal A} (hJ : IsSubDPIdeal hI J) : 
    DividedPowers J where
  dpow n x := if x ∈ J then hI.dpow n x else 0
  ...
\end{lstlisting}

The next three lemmas can be used to determine whether certain ideals in $A$ are sub-dp-ideals of $I$.
If $J$ is any ideal in $A$, it is not always the case that the intersection $J \cap I$ is a sub-dp-ideal; indeed, this requires the divided power structure on $I$ to satisfy the following compatibility condition modulo $J$.

\begin{lemma}
Given an ideal $J$ of $A$, the ideal $J \cap I$ is a sub-dp-ideal of $(I, \gamma)$ if and only if for every $n \in \N$ and every $a, b \in I$ such that $a - b \in J$, the element $\gamma_n(a) - \gamma_n(b)$ belongs to~$J$\href{https://github.com/mariainesdff/divided_powers_journal/blob/78f8956f06816221781a16b5223753c97b9b641c/DividedPowers/SubDPIdeal.lean#L178}{\extlink}.
\end{lemma}

\begin{lemma}
Let $S \subseteq I$ be a subset and let $J$ be the ideal generated by $S$. Then $J$ is a sub-dp-ideal of $I$ if and only if for all $n > 0$ and all $s \in S$, $\gamma_n(s)$ belongs to $J$. 
\end{lemma}

\begin{proof}
Necessity is obvious. The converse direction is shown by an induction argument on the generating set. See \cite[Lemma 3.6]{BO78} for a detailed proof and \lstinline{span_isSubDPIdeal_iff}\href{https://github.com/mariainesdff/divided_powers_journal/blob/78f8956f06816221781a16b5223753c97b9b641c/DividedPowers/SubDPIdeal.lean#L195}{\extlink} for our formalization.
\end{proof}

\begin{lemma}[{\cite[Proposition 1.6.1.(i)]{Ber74}}]
For every ideal $J$ of $A$, the ideal $I \cdot J$ is a sub-dp-ideal of $(I, \gamma)$\href{https://github.com/mariainesdff/divided_powers_journal/blob/78f8956f06816221781a16b5223753c97b9b641c/DividedPowers/SubDPIdeal.lean#L312}{\extlink}.
\end{lemma}

Next, we show that sub-dp-ideals of $(I, \gamma)$ form a complete lattice with respect to inclusion, which we register as an instance\href{https://github.com/mariainesdff/divided_powers_journal/blob/78f8956f06816221781a16b5223753c97b9b641c/DividedPowers/SubDPIdeal.lean#L404}{\extlink}. Note that this is only possible by using the bundled version of the definition of sub-dp-ideal.

\begin{lstlisting}
instance : CompleteLattice (SubDPIdeal hI) := by ...
\end{lstlisting}

In order to provide this instance, we first show that sub-ideals of $I$ form a complete lattice\href{https://github.com/mariainesdff/divided_powers_journal/blob/78f8956f06816221781a16b5223753c97b9b641c/DividedPowers/SubDPIdeal.lean#L95}{\extlink}; combining this with the fact that the forgetful function sending a sub-dp-ideal $J$ of $I$ to $J$ regarded as a sub-ideal is injective, we can then apply Mathlib's \lstinline{Function.Injective.completeLattice} to conclude our desired result.

The lattice of sub-ideals of $I$ is clearly bounded, with top element $I$ and bot element $(0)$. Recall that the supremum of two sub-ideals $J$ and $K$ of $I$ is equal to their sum $J + K$, and their infimum is their intersection $J \cap K$. After proving that sub-ideals form a complete lattice, to be able to carry out the outlined strategy we need to show that, if $J$ and $K$ are sub-dp-ideals of $(I, \gamma)$, then so are  these constructions (as well as the analogous claims for arbitrary families of sub-ideals). For instance, we prove in \lstinline{isSubDPIdeal_sup}\href{https://github.com/mariainesdff/divided_powers_journal/blob/78f8956f06816221781a16b5223753c97b9b641c/DividedPowers/SubDPIdeal.lean#L226}{\extlink} that the sum of two sub-dp-ideals is again a sub-dp-ideal, and in \lstinline{isSubDPIdeal_iSup}\href{https://github.com/mariainesdff/divided_powers_journal/blob/78f8956f06816221781a16b5223753c97b9b641c/DividedPowers/SubDPIdeal.lean#L235}{\extlink} the analogous statement for an indexed collection of sub-dp-ideals.

Note that there is a subtlety in defining the class instance \lstinline{InfSet (SubDPIdeal hI)}\href{https://github.com/mariainesdff/divided_powers_journal/blob/78f8956f06816221781a16b5223753c97b9b641c/DividedPowers/SubDPIdeal.lean#L370}{\extlink}
for the infimum of a set of sub-dp-ideals of $(I, \gamma)$: instead of simply defining the underlying set as the intersection $\bigcap J \in S$, we need to define it as $I \cap (\bigcap_{J\in S} J)$, so that when the set~$S$ is empty, this yields the desired result that the infimum is $I$ (instead of the whole ring $A$).
We also study the relation between divided power morphisms and sub-dp-ideals.

\begin{lemma}
Let $f : (A, I, \gamma) \to (B,J,\delta)$ be a divided power morphism. Then \textnormal{span}$(f(I))$ is a sub-dp-ideal of $J$\href{https://github.com/mariainesdff/divided_powers_journal/blob/78f8956f06816221781a16b5223753c97b9b641c/DividedPowers/SubDPIdeal.lean#L265}{\extlink} and $\ker f \cap I$ is a sub-dp-ideal of $I$\href{https://github.com/mariainesdff/divided_powers_journal/blob/78f8956f06816221781a16b5223753c97b9b641c/DividedPowers/SubDPIdeal.lean#L517}{\extlink}.
\end{lemma}

Finally, we wish to highlight a construction that will be useful for the formalization of the universal divided power algebra (see \S \ref{subsec:future}). Given a subset $S \subseteq I$, we define \code{SubDPIdeal.span}\href{https://github.com/mariainesdff/divided_powers_journal/blob/78f8956f06816221781a16b5223753c97b9b641c/DividedPowers/SubDPIdeal.lean#L446}{\extlink} as the smallest sub-dp-ideal of $I$ that contains $S$. We prove that this ideal is spanned by the family $\{\gamma_n(s)\}$ for $s \in S$ and $n > 0$\href{https://github.com/mariainesdff/divided_powers_journal/blob/78f8956f06816221781a16b5223753c97b9b641c/DividedPowers/SubDPIdeal.lean#L482}{\extlink}.
\begin{lstlisting}
def SubDPIdeal.span (S : Set A) : SubDPIdeal hI := 
  sInf {J : SubDPIdeal hI | S ⊆ J.carrier}
\end{lstlisting}

\subsection{Divided Powers on Quotients}
\label{subsec:quot}
In this section we discuss the existence and uniqueness of divided power structures on quotient rings, following \cite[Lemma 3.5]{BO78}.
Let $(A, I, \gamma)$ be a divided power ring and let $J$ be any ideal of $A$. We denote by $I\cdot A/J$ the ideal generated by the image of $I$ in the quotient ring $A/J$.

If there is a divided power structure on $I\cdot A/J$ such that the quotient map $A \to A/J$ is a divided power morphism, then $J \cap I$ is a sub-dp-ideal of $I$\href{https://github.com/mariainesdff/divided_powers_journal/blob/78f8956f06816221781a16b5223753c97b9b641c/DividedPowers/SubDPIdeal.lean#L589}{\extlink} (see Listing \ref{code:inf}).

\begin{lstlisting}[caption={$J \cap I$ as a sub-dp-ideal of $I$.},float=h, label=code:inf,captionpos=t,abovecaptionskip=-\medskipamount]
def DividedPowers.subDPIdeal_inf_of_quot {J : Ideal A}
    {hJ : DividedPowers (I.map (Ideal.Quotient.mk J))} {φ : DPMorphism hI hJ}
    (hφ : φ.toRingHom = Ideal.Quotient.mk J) :
    SubDPIdeal hI where
  carrier    := J ⊓ I
  ...
\end{lstlisting}

Conversely, if $J \cap I$ is a sub-dp-ideal of $I$, then there is a unique induced divided power structure on $I\cdot A/J$ such that the quotient map is a divided power morphism\href{https://github.com/mariainesdff/divided_powers_journal/blob/78f8956f06816221781a16b5223753c97b9b641c/DividedPowers/SubDPIdeal.lean#L733}{\extlink}. Note that we actually formalize this statement in the more general setting\footnote[3]{We mean `more general setting' in the type-theoretic sense. Mathematically, both are equivalent.} of a surjective ring homomorphism $f : A \to B$ such that $\ker f \cap I$ is a sub-dp-ideal of $I$, and then specialize the result to the quotient map $A \to A/J$.  Given such an $f : A \to B$, the divided power structure $\bar{\gamma}$ on $I\cdot B$ is defined by sending $n \in \N, f(a) \in f(I)$ to $f(\gamma_n(a))$\href{https://github.com/mariainesdff/divided_powers_journal/blob/78f8956f06816221781a16b5223753c97b9b641c/DividedPowers/SubDPIdeal.lean#L648}{\extlink}; note that by surjectivity of $B$, this is a function $\N \to \text{span}(f(I)) \to B$. We check that $\bar{\gamma}$ is the unique divided power structure on $\text{span}(f(I))$ such that the map $f$ is a divided power morphism from $I$ to $\text{span}(f(I))$\href{https://github.com/mariainesdff/divided_powers_journal/blob/78f8956f06816221781a16b5223753c97b9b641c/DividedPowers/SubDPIdeal.lean#L694}{\extlink}.

\begin{lstlisting}
def Quotient.OfSurjective.dividedPowers {f : A →+* B} 
    (hf : Function.Surjective f) (hIf : IsSubDPIdeal hI (ker f ⊓ I)) :
    DividedPowers (I.map f) where
  dpow n := Function.extend 
    -- If `b' is of the form `f a' for `a' in `I',
    (fun a ↦ f a : I → B)
    -- then `dpow n b' is given by `f (hI.dpow n a)';
    (fun a ↦ f (hI.dpow n a) : I → B) 
    -- otherwise, it is zero.
    0
  ...
  
def Quotient.dividedPowers (hIJ : hI.IsSubDPIdeal (J ⊓ I)) :
    DividedPowers (I.map (Ideal.Quotient.mk J)) :=     
  DividedPowers.Quotient.OfSurjective.dividedPowers hI   
    Ideal.Quotient.mk_surjective (IsSubDPIdeal_aux hI hIJ)

\end{lstlisting}

\section{The Exponential Module}
\label{sec:exp}

\subsection{The Topology on the Multivariate Power Series Ring}
\label{subsec:topology}

Let $A$ be a commutative ring and let $\sigma$ be a set;
we consider the ring $A\lbra \sigma\rbra $
of multivariate power series with coefficients
in~$A$ and indeterminates indexed by $\sigma$, and its subring $A[\sigma]$
of polynomials.
For clarity, one often writes $T_s$
for the indeterminate corresponding to $s\in \sigma$.
When $\sigma=\{X,Y\}$, say, a more classical notation would be $A\lbra X,Y\rbra$
for power series in~$X,Y$,
and $A[X,Y]$ for polynomials.

Monomials, i.e., finite products of indeterminates, can be 
realized as functions with finite support from~$\sigma$ to~$\N$,
and a power series in $A\lbra \sigma\rbra$ is nothing but a function from~$M$
to~$A$, where $M$ is the set of monomials.
For $m\in M$, one writes~$T^m$ for the finite
product $\prod_{s\in \sigma} T_s^{m_s}$.
We write $\coeff_m(f)$ for the coefficient of a monomial~$m$ in~$f$.
Polynomials correspond to those power series for which
all but finitely many coefficients are zero;
then $f$ is the finite sum $\sum_{m\in M} \coeff_m(f) T^m$.
When $f$ is a power series, this sum is infinite, and
hence it has no meaning a priori.

The Mathlib library provides the definition of multivariate\href{https://github.com/leanprover-community/mathlib4/blob/951bd16d55870f47b2e2cd98c794f933e29bd2b5/Mathlib/RingTheory/MvPowerSeries/Basic.lean#L83}{\extlink} and univariate\href{https://github.com/leanprover-community/mathlib4/blob/951bd16d55870f47b2e2cd98c794f933e29bd2b5/Mathlib/RingTheory/PowerSeries/Basic.lean#L55}{\extlink} power series rings, and develops some of their algebraic theory. However, before our work, the topological theory of power series rings was missing from the library.

If $A$ has a topology, the ring $A\lbra \sigma\rbra$ is naturally
induced with the product topology. In particular,
a sequence $(f_i)$
of power series converges to a power series~$f$ if and only if
for each monomial~$m$, the sequence $\coeff_m(f_i)$
converges to $\coeff_m(f)$\href{https://github.com/leanprover-community/mathlib4/blob/951bd16d55870f47b2e2cd98c794f933e29bd2b5/Mathlib/RingTheory/MvPowerSeries/PiTopology.lean#L90}{\extlink}.

The product topology on the multivariate power series ring is not declared as a global instance because there are other topologies that one might want to consider on this ring; for example, if $A$ is a normed ring, the supremum of the norm of the coefficients induces a topology on $A\lbra\sigma\rbra$ that is finer than the product topology. Consequently, we declare this product topology as a scoped instance that can be loaded by writing the instruction \code{open scoped MvPowerSeries.WithPiTopology}.

\begin{lstlisting}
scoped instance (σ A : Type*) [TopologicalSpace A] : 
    TopologicalSpace (MvPowerSeries σ A) :=
 Pi.topologicalSpace -- The product topology.
\end{lstlisting}

Note that, while in mathematical practice one typically considers the case where $A$ is a ring, the above definition only assumes that $A$ is a topological space. If $A$ is a topological ring,
meaning that the topology of~$A$ is such that the addition and multiplication
are continuous, then the product topology makes $A\lbra \sigma\rbra$ 
a topological ring as well\href{https://github.com/leanprover-community/mathlib4/blob/951bd16d55870f47b2e2cd98c794f933e29bd2b5/Mathlib/RingTheory/MvPowerSeries/PiTopology.lean#L111}{\extlink}. The analogous statement holds when $A$ is a topological semiring\href{https://github.com/leanprover-community/mathlib4/blob/951bd16d55870f47b2e2cd98c794f933e29bd2b5/Mathlib/RingTheory/MvPowerSeries/PiTopology.lean#L101}{\extlink}.

If $A$ is a Hausdorff topological space, then $A\lbra \sigma\rbra$ is Hausdorff\href{https://github.com/leanprover-community/mathlib4/blob/951bd16d55870f47b2e2cd98c794f933e29bd2b5/Mathlib/RingTheory/MvPowerSeries/PiTopology.lean#L75}{\extlink}. Hence, in this setting, we obtain for every power series~$f$
the equality $f=\sum_{m\in M} \coeff_m(f) T^m$, where the right hand side is well-defined
as a summable family\href{https://github.com/leanprover-community/mathlib4/blob/951bd16d55870f47b2e2cd98c794f933e29bd2b5/Mathlib/RingTheory/MvPowerSeries/PiTopology.lean#L193}{\extlink}.

If the ring $A$ is endowed with the discrete topology, then for $f \in A\lbra \sigma\rbra$ the limit
$\lim_{n \to \infty} f^n$ is equal to zero if and only if the constant coefficient of $f$ is nilpotent\href{https://github.com/leanprover-community/mathlib4/blob/951bd16d55870f47b2e2cd98c794f933e29bd2b5/Mathlib/RingTheory/MvPowerSeries/PiTopology.lean#L160}{\extlink}.

If $A$ has a uniform structure, then $A\lbra \sigma\rbra$
is naturally endowed with the product uniformity\href{https://github.com/leanprover-community/mathlib4/blob/951bd16d55870f47b2e2cd98c794f933e29bd2b5/Mathlib/RingTheory/MvPowerSeries/PiTopology.lean#L204}{\extlink}. If $A$ is a complete uniform space, then so is $A\lbra \sigma\rbra$\href{https://github.com/leanprover-community/mathlib4/blob/951bd16d55870f47b2e2cd98c794f933e29bd2b5/Mathlib/RingTheory/MvPowerSeries/PiTopology.lean#L215}{\extlink}.
Then $A[\sigma]$ is dense in~$A\lbra \sigma\rbra$\href{https://github.com/leanprover-community/mathlib4/blob/951bd16d55870f47b2e2cd98c794f933e29bd2b5/Mathlib/RingTheory/MvPowerSeries/Evaluation.lean#L109}{\extlink}.

\subsection{Linear Topologies on Multivariate Power Series Rings}
\label{subsec:linear-top}

Some ring topologies which are relevant in algebra have the particular property of being \emph{linear}, meaning that zero has a basis of open neighborhoods consisting of two-sided ideals. 

More generally, a topology on a left $A$-module $M$ is said to be \emph{linear} if the open left $A$-submodules of $M$ form a basis of neighborhoods of zero. We formalize this as a typeclass \code{IsLinearTopology}\href{https://github.com/leanprover-community/mathlib4/blob/951bd16d55870f47b2e2cd98c794f933e29bd2b5/Mathlib/Topology/Algebra/LinearTopology.lean#L98}{\extlink}.

\begin{lstlisting}[mathescape]
variable (A M : Type*) [Ring A] [AddCommGroup A] [Module A M] [TopologicalSpace M]

class IsLinearTopology where
  hasBasis_submodule' : ($\mathcal{N}$ (0 : M)).HasBasis
    (fun N : Submodule A M ↦ (N : Set M) ∈ $\mathcal{N}$ 0)
    (fun N : Submodule A M ↦ (N : Set M))
\end{lstlisting}

We show that the topology on the ring $A$ is linear (i.e., there is a basis of open neighborhoods of zero consisting of two-sided ideals) if and only if it is linear both when $A$ is regarded as an $A$-module and as an $A^\text{mop}$-module, where $A^\text{mop}$ denotes the multiplicative opposite of $A$\href{https://github.com/leanprover-community/mathlib4/blob/951bd16d55870f47b2e2cd98c794f933e29bd2b5/Mathlib/Topology/Algebra/LinearTopology.lean#L315}{\extlink}.

\begin{lstlisting}[mathescape]
theorem isLinearTopology_iff_hasBasis_twoSidedIdeal :
  IsLinearTopology A A ∧ IsLinearTopology A$^\text{mop}$ A ↔
    ($\mathcal{N}$ 0).HasBasis (fun I : TwoSidedIdeal A ↦ (I : Set A) ∈ $\mathcal{N}$ 0) 
       (fun I : TwoSidedIdeal A ↦ (I : Set A))
\end{lstlisting}

Considering this equivalence, instead of providing a special spelling for the property that the topology on the ring $A$ is linear, we write this in Lean as
\begin{lstlisting}[mathescape]
[IsLinearTopology A A] [IsLinearTopology A$^\text{mop}$ A]
\end{lstlisting}

If the topology on $A$ is linear, $I$ is an open ideal of~$A$ and $m$ is a monomial, then
the set~$I_m$ of power series $f\in A\lbra \sigma\rbra$ such 
that $\coeff_k(f)\in I$ for all monomials $k\leq m$
is an ideal of $A\lbra \sigma\rbra$. When $I$
runs over all open ideals of~$A$ and $m$ runs over all monomials,
this gives a basis of open neighborhoods of~$0$ in~$A\lbra \sigma\rbra$\href{https://github.com/leanprover-community/mathlib4/blob/951bd16d55870f47b2e2cd98c794f933e29bd2b5/Mathlib/RingTheory/MvPowerSeries/LinearTopology.lean#L48}{\extlink}. In other words, the topology on $A\lbra \sigma\rbra$ is linear as well\href{https://github.com/leanprover-community/mathlib4/blob/951bd16d55870f47b2e2cd98c794f933e29bd2b5/Mathlib/RingTheory/MvPowerSeries/LinearTopology.lean#L128}{\extlink}. In particular, this applies when $A$ is endowed with the discrete topology.

\subsection{Evaluation and Substitution of Multivariate Power Series}
\label{subsec:eval-subst}

Let $B$ be a topological commutative $A$-algebra  whose topology is linear.
For any $b\colon \sigma \to B$, one can evaluate polynomials $f\in A[\sigma]$
at $b$. This gives a ring morphism $\eval_b\colon A[\sigma]\to B$.
Contrary to what happens for polynomials, power series cannot always be evaluated at every point, at least if one wishes for good continuity properties. However, if $b_s \to 0$ in~$B$ along the filter of cofinite subsets,
and if for every $s\in \sigma$, the sequence $b_s^n$ converges to~$0$
(one says that $b_s$ is \emph{topologically nilpotent}),
then $\eval_b$ is uniformly continuous,
and extends uniquely to a continuous morphism from $A\lbra \sigma\rbra$ to~$B$.

We formalize the relevant conditions as a structure
\lstinline{MvPowerSeries.HasEval}\href{https://github.com/leanprover-community/mathlib4/blob/951bd16d55870f47b2e2cd98c794f933e29bd2b5/Mathlib/RingTheory/MvPowerSeries/Evaluation.lean#L63}{\extlink},

\begin{lstlisting} [mathescape, caption={Conditions for evaluation of power series being defined at $b : \sigma \to B$.},float=h, label=code:HasEval,captionpos=t,abovecaptionskip=-\medskipamount]
structure HasEval (b : σ → B) : Prop where
  hpow : ∀ s, IsTopologicallyNilpotent (b s)
  tendsto_zero : Tendsto b cofinite ($\mathcal{N}$ 0)
\end{lstlisting}

\removelastskip\noindent
in which we use our definition \code{IsTopologicallyNilpotent}\href{https://github.com/leanprover-community/mathlib4/blob/951bd16d55870f47b2e2cd98c794f933e29bd2b5/Mathlib/Topology/Algebra/TopologicallyNilpotent.lean#L42}{\extlink}:

\begin{lstlisting}[mathescape]
def IsTopologicallyNilpotent {A : Type*} [MonoidWithZero A] 
    [TopologicalSpace A] (a : A) : Prop :=
  Tendsto (a ^ ·) atTop ($\mathcal{N}$ 0) -- The powers of `a' converge to `0'.
\end{lstlisting}

We provide the definition \code{MvPowerSeries.aeval}\href{https://github.com/leanprover-community/mathlib4/blob/951bd16d55870f47b2e2cd98c794f933e29bd2b5/Mathlib/RingTheory/MvPowerSeries/Evaluation.lean#L296}{\extlink} to evaluate a multivariate power series $f \in A\lbra \sigma \rbra$ at a point $b : \sigma \to B$. This agrees with the evaluation of $f$ as a polynomial whenever $f \in A[\sigma]$. Otherwise, it is defined by density from polynomials; its values are irrelevant unless the algebra map $A \to B$ is continuous and $b$ satisfies the two conditions
bundled in \code{MvPowerSeries.HasEval b}. The function \code{MvPowerSeries.aeval} is an $A$-algebra map, whose underlying function and ring homomorphism are denoted \color{black}\code{MvPowerSeries.eval₂}\href{https://github.com/leanprover-community/mathlib4/blob/951bd16d55870f47b2e2cd98c794f933e29bd2b5/Mathlib/RingTheory/MvPowerSeries/Evaluation.lean#L193}{\extlink} and \code{MvPowerSeries.eval₂Hom}\href{https://github.com/leanprover-community/mathlib4/blob/951bd16d55870f47b2e2cd98c794f933e29bd2b5/Mathlib/RingTheory/MvPowerSeries/Evaluation.lean#L220}{\extlink}, respectively (see Listing \ref{code:eval}).

\begin{lstlisting}[caption={Evaluation of multivariate power series.},float=h, label=code:eval,captionpos=t,abovecaptionskip=-\medskipamount]
def eval₂ (f : MvPowerSeries σ A) (φ : A →+* B) (b : σ → B): B :=
  if H : ∃ p : MvPolynomial σ A, p = f then (MvPolynomial.eval₂ φ b H.choose)
  else IsDenseInducing.extend coeToMvPowerSeries_isDenseInducing (MvPolynomial.eval₂ φ b) f
 
-- eval₂ as a ring homomorphism
def eval₂Hom (hφ : Continuous φ) (hb : HasEval b) := ... 
  
--Specialize `eval₂' to the case where `φ' is `AlgebraMap A B'.
def aeval (hb : HasEval b) : MvPowerSeries σ A →ₐ[A] B where
  toRingHom := MvPowerSeries.eval₂Hom (continuous_algebraMap A B) hb
  ...
\end{lstlisting}
We specialize our work to provide a corresponding \code{PowerSeries.aeval} definition\href{https://github.com/leanprover-community/mathlib4/blob/951bd16d55870f47b2e2cd98c794f933e29bd2b5/Mathlib/RingTheory/PowerSeries/Evaluation.lean#L159}{\extlink} for evaluation of univariate power series.

Under certain conditions, it is possible to substitute multivariate power series within other power series. This is a special case of evaluation of power series, in which the ground ring $A$ is given the discrete topology. Indeed, let $\tau$ be a set and consider a map $b \colon \sigma \to B\lbra\tau\rbra$ that assigns to each indeterminate in $\sigma$ a power series in $B\lbra\tau\rbra$.
Provided that the constant coefficient of $b_s$ is nilpotent for each $s \in \sigma$, and that $b_s \to 0$ in~$B\lbra\tau\rbra$ along the filter of cofinite subsets, then we can substitute $b$ into a power series $f \in A\lbra\sigma\rbra$. We record these conditions in a structure \code{MvPowerSeries.HasSubst}\href{https://github.com/leanprover-community/mathlib4/blob/8a6f6f7dbe193f0f8a6c0a33e3430e70748f96cf/Mathlib/RingTheory/MvPowerSeries/Substitution.lean#L71}{\extlink}.
\begin{lstlisting}
structure HasSubst (a : σ → MvPowerSeries τ S) : Prop where
  const_coeff s : IsNilpotent (constantCoeff τ S (a s))
  coeff_zero d : {s | (a s).coeff S d ≠ 0}.Finite
\end{lstlisting}

Note that if the set $\sigma$ is finite, to check that $b$ can be substituted into $f$, it is enough to check the nilpotency condition on the coefficients\href{https://github.com/leanprover-community/mathlib4/blob/8a6f6f7dbe193f0f8a6c0a33e3430e70748f96cf/Mathlib/RingTheory/MvPowerSeries/Substitution.lean#L137}{\extlink}. This is in particular true if the constant coefficients of the $b_s$ are zero\href{https://github.com/leanprover-community/mathlib4/blob/8a6f6f7dbe193f0f8a6c0a33e3430e70748f96cf/Mathlib/RingTheory/MvPowerSeries/Substitution.lean#L144}{\extlink}, which is often the case in practical applications.

Substitution of multivariate power series is then defined in \code{MvPowerSeries.subst}\href{https://github.com/leanprover-community/mathlib4/blob/8a6f6f7dbe193f0f8a6c0a33e3430e70748f96cf/Mathlib/RingTheory/MvPowerSeries/Substitution.lean#L153}{\extlink}. As happened with the evaluation, the substitution is only well defined for maps $b \colon \sigma \to B\lbra\tau\rbra$ satisfying \code{HasSubst b}. We provide an extensive API to work with this definition, describing its compatibility with multiple algebraic operations and providing formulas for the coefficients of the substitution. We also provide the specialization of this definition to the case where $f$ is a univariate power series\href{https://github.com/leanprover-community/mathlib4/blob/8a6f6f7dbe193f0f8a6c0a33e3430e70748f96cf/Mathlib/RingTheory/PowerSeries/Substitution.lean#L145}{\extlink}.

\begin{lstlisting}
variable (b : σ → MvPowerSeries τ B) (f : MvPowerSeries σ A)

def MvPowerSeries.subst : MvPowerSeries τ B :=
  letI : UniformSpace A := ⊥ -- discrete uniformity on `A'
  letI : UniformSpace B := ⊥ -- discrete uniformity on `B'
  MvPowerSeries.eval₂ (algebraMap _ _) b f

-- `MvPowerSeries.subst' as an `A'-algebra morphism.
def MvPowerSeries.substAlgHom (hb : HasSubst b) :
    MvPowerSeries σ R →ₐ[R] MvPowerSeries τ S := ...
\end{lstlisting}

\subsection{The Exponential Module}
\label{subsec:exp}

Let $A$ be a ring. A power series $f\in A\lbra X\rbra$ in one indeterminate is said to be \emph{of exponential type} if $f(0)=1$ and if, substituting~$X$ for~$X_0+X_1$ in~$f$, one has the functional relation $f(X_0+X_1)=f(X_0)\cdot f(X_1)$.
We formalize this definition as a structure \code{IsExponential}\href{https://github.com/mariainesdff/divided_powers_journal/blob/78f8956f06816221781a16b5223753c97b9b641c/DividedPowers/ExponentialModule/Basic.lean#L462}{\extlink}.
\begin{lstlisting}[mathescape]
structure IsExponential (f : $A\lbra X\rbra$) : Prop where
  add_mul : subst (S := A) ($X_0$ + $X_1$) f = subst $X_0$ f * subst $X_1$ f
  constantCoeff : constantCoeff A f = 1
\end{lstlisting}

Let $\mathscr E(A)$ be the set of power series of exponential type. For $f,g\in\mathscr E(A)$, one has $f g\in\mathscr E(A)$. Moreover, for every $a\in A$ and every $f\in\mathscr E(A)$, the rescaled power series $f(aX)$ also belongs to~$\mathscr E(A)$. Because of that, the set $\mathscr E(A)$ is an $A$-module:
the addition law of the module is the product of power series, and the external law is given by rescaling.

There is a subtlety to consider in the formalization of the definition of $\mathscr E(A)$. The functional equation $f(X_0+X_1)=f(X_0)\cdot f(X_1)$ satisfied by an exponential power series relates an additive expression to a multiplicative one, so its formalization requires combining additive and multiplicative structures. 
Mathlib provides the definition \code{Additive} to deal with the required conversion: if \code{M} is a type that carries a multiplicative structure, then \code{Additive M} is a new type in bijection with \code{M} that carries the corresponding additive structure. The bijection \code{ofMul : M → Additive M} satisfies \code{ofMul(x * y) = ofMul x + ofMul y} for all \code{x y : M}; its inverse is called \code{toMul}.

We formalize the definition of $\mathscr E(A)$ as an additive submonoid of the type \code{Additive} $A\lbra X\rbra$, which we call \code{ExponentialModule A}\href{https://github.com/mariainesdff/divided_powers_journal/blob/78f8956f06816221781a16b5223753c97b9b641c/DividedPowers/ExponentialModule/Basic.lean#L638}{\extlink}. The underlying set of this submonoid consists of the terms \code{f : Additive} $A\lbra X\rbra$ such that the corresponding power series \code{toMul f :} $A\lbra X\rbra$ is exponential.

\begin{lstlisting}[mathescape]
def ExponentialModule : AddSubmonoid (Additive $A\lbra X\rbra$) where
  carrier := { f : Additive ($A\lbra X\rbra$) | IsExponential (toMul f) }
  add_mem' {f g} hf hg := by simp only [Set.mem_setOf_eq, toMul_add, hf.mul hg]
  zero_mem' := by simp only [Set.mem_setOf_eq, toMul_zero, IsExponential.one]
\end{lstlisting}

The instances \code{instAddCommGroup}\href{https://github.com/mariainesdff/divided_powers_journal/blob/78f8956f06816221781a16b5223753c97b9b641c/DividedPowers/ExponentialModule/Basic.lean#L709}{\extlink} and \code{instModule}\href{https://github.com/mariainesdff/divided_powers_journal/blob/78f8956f06816221781a16b5223753c97b9b641c/DividedPowers/ExponentialModule/Basic.lean#L686}{\extlink} show that \code{ExponentialModule A} is both a commutative additive group and an $A$-module. We also prove that if $f$ is an exponential power series, then $f$ is invertible\href{https://github.com/mariainesdff/divided_powers_journal/blob/78f8956f06816221781a16b5223753c97b9b641c/DividedPowers/ExponentialModule/Basic.lean#L573}{\extlink}, and its inverse, given by the power series $f(-X)$\href{https://github.com/mariainesdff/divided_powers_journal/blob/78f8956f06816221781a16b5223753c97b9b641c/DividedPowers/ExponentialModule/Basic.lean#L576}{\extlink}, is also of exponential type\href{https://github.com/mariainesdff/divided_powers_journal/blob/78f8956f06816221781a16b5223753c97b9b641c/DividedPowers/ExponentialModule/Basic.lean#L587}{\extlink}.

If $(A,I,\gamma)$ is a divided power ideal, then for every $a\in I$, the power series $\exp_I(aX)=\sum\gamma_n(a) X^n$ is of exponential type\href{https://github.com/mariainesdff/divided_powers_journal/blob/78f8956f06816221781a16b5223753c97b9b641c/DividedPowers/Exponential.lean#L35}{\extlink}.
To give an intuitive explanation for this property, recall that for every real number $a$, $\exp(ax)$ is given by the power series $\sum a^n/n! \cdot x^n$.  Moreover, the functional equation of the exponential function implies that $\exp(a (x+y))=\exp(ax) \exp(ay)$. In the algebraic setting, divided powers emulate the functions $a \mapsto a^n/n!$ and the proof  that $\exp_I(aX)$ is of exponential type is established by manipulations of formal power series completely similar to those that establish the functional equation of the exponential power series in analysis.
The map $a\mapsto \exp_I(aX)$
from~$I$ to~$\mathscr E(A)$ is a morphism of $A$-modules\href{https://github.com/mariainesdff/divided_powers_journal/blob/78f8956f06816221781a16b5223753c97b9b641c/DividedPowers/Exponential.lean#L54}{\extlink}.

\begin{lstlisting}[mathescape]
def exp (hI : DividedPowers I) (a : A) : PowerSeries A :=
  PowerSeries.mk fun n $\mapsto$ hI.dpow n a
  
def exp' (hI : DividedPowers I) (a : I) : ExponentialModule A :=
  ⟨hI.exp a, hI.exp_isExponential a.prop⟩
  
def exp'_linearMap (hI : DividedPowers I) : I $\to_l$[A] ExponentialModule A where
  toFun := hI.exp'
\end{lstlisting}

\subsection{Application: Divided Powers on a Sum of Ideals}
\label{subsec:ideal-add}

If $(I, \gamma_I)$ and $(J, \gamma_J)$ are two divided power $A$-ideals with divided powers that agree on $I \cap J$, then the ideal $I + J$ has a unique divided power structure $\gamma_{I + J}$ extending those on $I$ and $J$. The starting point for the construction of this divided power structure
is the construction of the linear map $I+J\to\mathscr E(A)$ that extends
the two linear maps $I\to\mathscr E(A)$ and $J\to\mathscr E(A)$ associated with the divided power structures on these ideals.

\begin{lstlisting}[mathescape, float=th]
def IdealAdd.exp'_linearMap (hIJ : ∀ n, ∀ a ∈ I ⊓ J, hI.dpow n a = hJ.dpow n a) : 
    (I + J) →$_l$[A] (ExponentialModule A) := 
  LinearMap.onSup (f := hI.exp'_linearMap) (g := hJ.exp'_linearMap)
    (fun x hxI hxJ ↦ Subtype.coe_inj.mp
      (Additive.toMul.injective (PowerSeries.ext (fun _ ↦ hIJ x ⟨hxI, hxJ⟩))))
\end{lstlisting}
In order to formalize this definition, we first formalize the more general construction \code{LinearMap.onSup}\href{https://github.com/mariainesdff/divided_powers_journal/blob/78f8956f06816221781a16b5223753c97b9b641c/DividedPowers/ForMathlib/LinearAlgebra/OnSup.lean#L59}{\extlink}, which given two $A$-submodules $M, N$ of an $A$-module $B$, and two $A$-linear maps $f : M \to B$ and $g : N \to B$ that agree on $M \cap N$, is the unique $A$-linear map $M + N \to B$ that simultaneously extends $f$ and $g$.

The divided power structure on $I + J$ is then defined by extending \code{IdealAdd.exp'_linearMap} by zero\href{https://github.com/mariainesdff/divided_powers_journal/blob/78f8956f06816221781a16b5223753c97b9b641c/DividedPowers/IdealAdd.lean#L94}{\extlink}.
\begin{lstlisting}
def IdealAdd.dpow (n : ℕ) := Function.extend 
  -- If `a' is in `I + J',
  (fun a ↦ ↑a : (I + J) → A) 
  -- then `dpow n a' is the nth coefficient of `exp_{I + J} (a)'
  (fun a ↦ coeff A n (exp'_linearMap hIJ a))
  --otherwise, it is zero
  0
\end{lstlisting}

This furnishes the divided power map and also takes care
of the axioms~(i-v) of Definition~\ref{def:dp}.
The verification of the two remaining axioms is essentially
formal. For axiom~(vii), it requires to check a
combinatorial property which is itself proved using 
the theory of divided powers on the $\Q$-algebra $\Q[X]$
of polynomials\href{https://github.com/mariainesdff/divided_powers_journal/blob/78f8956f06816221781a16b5223753c97b9b641c/DividedPowers/IdealAdd.lean#L351}{\extlink}.

Since $\gamma_{I+ J}$ is by construction compatible with $\gamma_I$ and $\gamma_J$, it follows that the identity map on $A$ provides two divided power morphisms $\text{id}_A : (A, I, \gamma_I) \to (A, I +J, \gamma_{I + J})$\href{https://github.com/mariainesdff/divided_powers_journal/blob/78f8956f06816221781a16b5223753c97b9b641c/DividedPowers/IdealAdd.lean#L461}{\extlink}  and
$\text{id}_A : (A, J, \gamma_J) \to (A, I +J, \gamma_{I + J})$\href{https://github.com/mariainesdff/divided_powers_journal/blob/78f8956f06816221781a16b5223753c97b9b641c/DividedPowers/IdealAdd.lean#L465}{\extlink}. Similarly, the compatibility condition implies that both $(I, \gamma_I)$\href{https://github.com/mariainesdff/divided_powers_journal/blob/78f8956f06816221781a16b5223753c97b9b641c/DividedPowers/IdealAdd.lean#L480}{\extlink} and $(J, \gamma_J)$\href{https://github.com/mariainesdff/divided_powers_journal/blob/78f8956f06816221781a16b5223753c97b9b641c/DividedPowers/IdealAdd.lean#L490}{\extlink} are sub-dp-ideals of $(I + J, \gamma_{I+ J})$ .

The fact that $\gamma_{I+ J}$ is the unique divided power structure on $I + J$ which is compatible with both $\gamma_I$ and $\gamma_J$ is proven in the theorem \code {DividedPowers.IdealAdd.dividedPowers_unique}\href{https://github.com/mariainesdff/divided_powers_journal/blob/78f8956f06816221781a16b5223753c97b9b641c/DividedPowers/IdealAdd.lean#L447}{\extlink}.

\section{Discussion}
\label{sec:discuss}

\subsection{Remarks about the Implementation}
\label{subsec:impl}
\subsubsection{The Type of the \texttt{dpow} Map}
\label{subsubsec:dpow}
In Definition \ref{def:dp}, a divided power structure on an ideal $I$ of a ring $A$ is a family of maps $\{\gamma_n : I \to A\}_{n \in \N}$ satisfying properties (i) through (vii). The most literal representation of this family in Lean would be as a function \lstinline{dpow : ℕ → I → A}. However, in the structure \lstinline{DividedPowers}, we instead chose to represent  this family as a map \lstinline{dpow : ℕ → A → A} satisfying (i) through (vii) plus the condition that \lstinline{dpow n x = 0} for all $n \in \N$ and all $x \not\in I$.

By opting for the second definition, we avoid having to provide explicit conversions between types. For example, to define divided powers on the sum of two ideals $I$ and $J$ of $A$, we need to impose the following compatibility condition on the elements of the intersection of these ideals. 
\begin{lstlisting}
∀ (n : ℕ) (a : A), ∀ (ha : a ∈ I ⊓ J), hI.dpow n a = hJ.dpow n a
\end{lstlisting}
If \lstinline{dpow} were instead defined as a map \lstinline{ℕ → I → A}, this would have to be rewritten as
\begin{lstlisting}
∀ (n : ℕ) (a : A), ∀ (ha : a ∈ I ⊓ J), hI.dpow n ⟨a, haI⟩ = hJ.dpow n ⟨a, haJ⟩
\end{lstlisting}
where \lstinline{haI} and \lstinline{haJ} are proofs that $a$ belongs to the ideals $I$ and $J$, respectively.

\subsubsection{Divided Powers on Sums of Ideals}
\label{subsubsec:ideal-add-compare}

If $(I,\gamma_I)$ and $(J,\gamma_J)$ are two divided power ideals of a ring~$A$
such that $\gamma_I$ and $\gamma_J$ coincide on $I\cap J$,
the construction of the divided power map~$\gamma_{I + J}$ on $I+J$
combines the two  $A$-linear maps $I\to \mathscr E(A)$
and $J\to \mathscr E(A)$ to a linear map on~$I+J$.

This approach, while mathematically satisfying, posed several difficulties:
\begin{itemize}
\item
The most important one was the necessity of formalizing the construction
of the exponential module $\mathscr E(A)$, to define evaluation and substitution
of power series. Despite its length, this \emph{detour} brings a worthwhile
addition to the Mathlib library.
\item 
A second one was the absence in the Mathlib library of 
the definitions that construct, given two modules $M$ and~$N$,
two submodules~$P$ and~$Q$ of~$M$ and two linear maps $f\colon P\to N$ 
and $g\colon Q\to N$ that agree on $P\cap Q$, 
a linear map $h\colon P+Q\to N$ that extends~$f$ and~$g$\href{https://github.com/mariainesdff/divided_powers_journal/blob/78f8956f06816221781a16b5223753c97b9b641c/DividedPowers/ForMathlib/LinearAlgebra/OnSup.lean#L59}{\extlink}.
\item
A third one is that this construction only defines the divided power map
on the ideal $I+J$, and one still has to extend it by~$0$.
\end{itemize}

In a first implementation, our construction of the divided power structure on the sum
of two ideals had been more elementary.
It consisted in proving that for $x\in I$ and $y\in J$,
the expression $\sum_{k=0}^n \gamma_{I, k}(x)\gamma_{J, n-k}(y)$
only depends on $x+y$\href{https://github.com/mariainesdff/divided_powers_journal/blob/78f8956f06816221781a16b5223753c97b9b641c/DividedPowers/IdealAdd_v1.lean#L55}{\extlink}; call it $\gamma_{I + J, n}(x+y)$.
Since any element of $I+J$ can be decomposed in this manner,
this defines a function~$\gamma_{I + J,n}$ on $I+J$\href{https://github.com/mariainesdff/divided_powers_journal/blob/78f8956f06816221781a16b5223753c97b9b641c/DividedPowers/IdealAdd_v1.lean#L48}{\extlink}, 
and it remains to prove that the axioms of a divided power structure
are satisfied.

\begin{lstlisting}[mathescape]
def IdealAdd_v1.dpow : ℕ → A → A := fun n => Function.extend 
  -- if `c' is of the form `a + b', for `a' in `I' and `b' in `J',
  (fun ⟨a, b⟩ => (a : A) + (b : A) : I × J → A)
  -- then `dpow n c' is given by $\sum_{k=0}^n \gamma_{I, k}(x)\gamma_{J, n-k}(y)$
  (fun ⟨a, b⟩ => Finset.sum (range (n + 1)) fun k => 
    hI.dpow k (a : A) * hJ.dpow (n - k) (b : A))
  -- otherwise it is 0.
  0
\end{lstlisting}

For comparison purposes, we provide both versions of the implementation of divided powers on sums of ideals: our preferred implementation using the exponential map is in the file \code{IdealAdd.lean}\href{https://github.com/mariainesdff/divided_powers_journal/blob/78f8956f06816221781a16b5223753c97b9b641c/DividedPowers/IdealAdd.lean}{\extlink}, while the older implementation is in \code{IdealAdd_v1.lean}\href{https://github.com/mariainesdff/divided_powers_journal/blob/78f8956f06816221781a16b5223753c97b9b641c/DividedPowers/IdealAdd_v1.lean}{\extlink}.

\subsubsection{Evaluation and Substitution of Power Series}
\label{subsubsec:eval}

Recall that the evaluation of a power series $f \in A\lbra\sigma\rbra$ at a point $b \colon \sigma \to B$ is only mathematically relevant when $b$ satisfies the conditions recorded in the structure \code{MvPowerSeries.HasEval} (see Listing \ref{code:HasEval}).

However, to formalize the evaluation in Lean,
it is more practical to extend it to all $b \colon \sigma \to B$ by providing ``dummy'' values outside the relevant domain of definition. As usual, we choose the value zero, except when \lstinline{f} is (the coercion of) a polynomial,
in which case we can provide the natural evaluation of that polynomial. Since substitution of power series is a special case of evaluation, the same remark applies.

The downside is that the function thus obtained has weak algebraic properties, and this makes
the verification of functional equations somewhat complicated. For instance, when proving theorems that involve nested substitutions, 
one needs to check that the \code{HasSubst} condition is satisfied for all relevant power series. Since these verifications are often mathematically obvious, and performed in a similar way, it would be interesting to implement tactics that automate these tasks.

\subsubsection{Inverting Factorials}
\label{subsubsec:inverse}
We showed in Section \ref{subsec-dp-ex} that, whenever $I$ is an ideal in a commutative ring $A$ and $n$ is a natural number such that $(n-1)!$ is invertible in $A$ and $I^n = 0$,
the family $\gamma_m(x) = \frac{x^m}{m!}$ is a divided power structure on $I$. This is true in particular in the three cases described in Examples \ref{ex:dp-nil}.

The proofs of these examples require to manipulate inverses of factorial elements. For instance, an intermediate step in the proof of \code{dpow_mul_of_add_lt} is to check that the following equality holds in A:
\begin{lstlisting}
Ring.inverse m ! * Ring.inverse k ! = ((m + k).choose m) * Ring.inverse (m + k)!
\end{lstlisting}

In this situation, $m!$, $k!$ and $(m + k)!$ are all invertible in $A$, so an easy algebraic manipulation yields the following equivalent equality.
\begin{lstlisting}
↑(m + k)! = ↑k ! * (↑m ! * ↑((m + k).choose m))
\end{lstlisting}

However, currently this manipulation has to be performed manually. It would be useful to implement a tactic that can remove uses of \code{Ring.inverse} as in the example above (creating side goals to prove the required invertibility hypotheses). The resulting expressions could then, at least in certain cases, be further simplified by using tactics like \code{ring} (which does not deal with inversion) or \code{norm_cast}.

\subsection{Future Work}
\label{subsec:future}

\subsubsection{Divided Powers on Discrete Valuation Rings}\label{par:dvr}
Example \ref{ex:padic} can be generalized to construct divided powers on discrete valuation rings. Indeed, let $R$ be a discrete valuation ring of mixed characteristic $p$ and uniformizer $\pi$ and write $p = u \cdot \pi^e$, where $u \in R$ is a unit and $e \in \N$. Then the ideal $(\pi)$ has divided powers if and only if $e < p$ \cite[\S3, Examples 3]{BO78}. This result relies on the theory of extensions of valuations to ring extensions, which has been formalized in \cite{dFFNMMdC} but is not yet available in Mathlib. 

\subsubsection{The Divided Power Envelope}\label{par:dp-env}

If $(A, I, \gamma)$ is a divided power ring and $J \subseteq B$ is an ideal in an $A$-algebra $B$ such that $I \cdot B \subseteq J$, then there exists a divided power $A$-algebra, called the \emph{divided power envelope of $J$ relative to $(A, I, \gamma)$} and denoted by $D_{B, \gamma}(J)$, with the following universal property: for every divided power morphism $(A, I, \gamma) \to (C, K, \delta)$, there is a one-to-one correspondence between ring homomorphisms $B \to C$ that map $J$ into $K$ and $A$-linear divided power morphisms $D_{B, \gamma}(J) \to (C, K, \delta)$ \cite[Theorem 3.19]{BO78}.

The construction of the divided power envelope relies on a second universal construction, called the \emph{universal divided power algebra} \cite{Roby63, Roby65}. When $M$ is a module over a ring $A$, the universal divided power algebra of $M$ is a graded ring $\Gamma_A (M)$, together with an $A$-linear map $M\to \Gamma_A^1(M)$, whose augmentation ideal $\Gamma_A^+(M)$ has a divided power structure, and which is universal for these properties.

The construction of $\Gamma_A(M)$ as a plain graded algebra, together with the map $M\to \Gamma_A^1(M)$, is relatively straightforward, but it is a difficult theorem of Roby \cite{Roby65} that the augmentation ideal has divided powers. In fact, we discovered a gap in his proof during the formalisation process, which we were able to solve by providing an alternative argument \cite{CLdFF}.

We plan to conclude the formalisation of this construction and that of the envelope in the forthcoming months.
Note that the complete informal proof of the existence of these universal objects expands dozens of pages, with most of the work devoted to showing that the universal divided power algebra of a free module has free graded components.

\subsubsection{The Rings of $p$-adic Periods}\label{par:Hodge}

Once the divided power envelope has been formalized, we will use it to formalize the 
construction of the Fontaine period ring $B_{\crys}$. The definition of this ring is highly technical, and it involves several steps which include defining the $p$-adic complex numbers \cite{dFF23}, taking a ring of Witt vectors \cite{CL21} and constructing the divided power envelope with respect to a certain ring homomorphism.

Fontaine period rings, including $B_{\crys}$, are a fundamental tool in $p$-adic Hodge theory, an active area of research in number theory devoted to the study of $p$-adic Galois representations. Besides being used to detect interesting properties of representations, these period rings are also prominently used in comparison
theorems between different cohomology theories; for instance, the ring $B_{\crys}$ is used in a comparison theorem between \'etale cohomology and crystalline cohomology of $p$-adic varieties. Hence, the formalization of these rings will open the door to formalizing key research in number theory and algebraic geometry.

In particular, as mentioned in the introduction, the work described in this paper is a prerequisite for the formalization of a complete proof of Fermat's Last Theorem. More concretely, this proof relies on an ``$R = T$'' isomorphism between two deformations rings, where the ring $R$ corresponds to crystalline deformations of a given mod $p$-representation.

The ring $B_\crys$ is used to identify crystalline representations; hence, while the formalization of the full theory of crystalline cohomology is not required for the proof of Fermat's Last Theorem, the formalization of the divided power envelope certainly is.



\bibliography{divided_powers}

\end{document}